\documentclass[journal]{IEEEtran}
\usepackage{amsmath,amssymb,amsthm}
\usepackage{booktabs}
\usepackage{hyperref}
\usepackage{cite}
\usepackage{graphicx}
\graphicspath{{figures/}}

\newtheorem{definition}{Definition}
\newtheorem{proposition}{Proposition}
\newtheorem{remark}{Remark}

\newtheorem{claim}{Claim}

\begin{document}

\title{Stateful Pricing and Allocation for Repeated Constrained DER Coordination
in Distribution Networks}

\author{Shaun~Sweeney,~\IEEEmembership{}
        Peter~Kilby,~\IEEEmembership{}
        Blake~Penney,~\IEEEmembership{}
        Komeil Moghaddasi,~\IEEEmembership{}
        Sunera Mudiyanselage~\IEEEmembership{}
\thanks{S. Sweeney is with the Dyson School of Design Engineering, Imperial College
London, London, U.K. (e-mail: s.sweeney21@imperial.ac.uk).}
\thanks{P. Kilby is with Energy Queensland, Brisbane, Australia
(e-mail: peter.kilby@energyq.com.au).}
\thanks{B. Penney is with Black Japan Group, Melbourne, Australia.}
\thanks{K. Moghaddasi and S. Mudiyanselage are with Queensland University of
Technology, Australia.}
\thanks{Manuscript received June 2026.}
}

\maketitle

\begin{abstract}
Distribution networks with high penetrations of distributed energy resources (DERs)
must repeatedly allocate limited network capability in two directions: under import
scarcity, which flexible demand is served; under export congestion, which generation is
curtailed. Dynamic operating envelopes (DOEs) enforce hard feasibility bounds but lack
intertemporal correction; dynamic network prices (DNPs) provide an allocative signal
but cannot guarantee constraint satisfaction. This paper develops a stateful
cyber-physical coordination mechanism---termed an Automatic Market Maker (AMM)---as an
additive coordination layer for machine-to-machine DER access, combining dual fairness
states for import and export, bounded bilateral prices driven by a voltage-aware deficit
signal, and feasibility-constrained matching within a two-tier MV/LV architecture.
Experiments on the CSIRO MV+33LV feeder dataset compare five mechanisms and
additionally benchmark the fair-over-time DOE formulations of Moring
\emph{et al.}~\cite{moring2024} (FET, FOT, FUH) on the same dataset. Relative to
equal-allocation DOE, the AMM reduces unserved flexible demand by 76\% (96.0~MWh to
23.2~MWh) with zero thermal violations and reduces export curtailment (85.4~MWh to
64.5~MWh). Near-identical DOE and DOE-GREEDY performance confirms that heuristic
choice alone does not improve repeated constrained outcomes. The AMM reaches
inter-feeder Jain index 0.9998 annually, outperforming all DOE variants from month~6
onwards. Direct benchmarking against FET/FOT/FUH shows those mechanisms achieve higher
worst-feeder equity (0.989--0.990 versus AMM's 0.914) through an explicit max-min MV
objective, but operate offline over predetermined horizons and do not provide bilateral
scarcity signals, real-time operation, or import-side participant-level intertemporal
correction. The AMM and FET/FOT/FUH address different operational objectives and are
potentially composable; their combination is identified as a primary direction for
future work.
\end{abstract}


\section{Introduction}
\label{sec:intro}

\subsection{Motivation}

Distribution networks with high penetrations of rooftop photovoltaic (PV), household
batteries, and electric vehicles (EVs) increasingly require active coordination of
time-coupled demand and generation. Two recurrent allocation problems arise. Under
import scarcity, the operational question is which flexible demand is served when
available import capability is insufficient to satisfy all requests. Under export
congestion, the corresponding question is which generation is curtailed when local
export capability is limited. These two cases differ operationally but share the same
structural feature: a constrained physical resource must be allocated repeatedly over
time.

At the same time, distribution-level coordination is becoming increasingly automated.
Machine-to-machine (M2M) communication enables devices and aggregators to observe
signals and update behaviour at high frequency. This capability is operationally useful,
but it also changes the strategic properties of the coordination protocol. Allocation
rules that are acceptable under slow or human-mediated participation may become
exploitable under automated, repeated interaction.

These observations motivate five coordination requirements for active distribution
networks: (1)~\emph{feasibility}: network constraints must be respected;
(2)~\emph{scarcity signalling}: participants should observe system tightness and adapt;
(3)~\emph{allocation}: limited headroom should be distributed by a transparent rule;
(4)~\emph{intertemporal correction}: repeated disadvantage must be visible and
correctable across time; and (5)~\emph{strategic robustness}: the protocol should not
create obvious gains from high-frequency timing or misreporting under M2M participation.

Current mechanisms address these requirements only partially. DOEs provide strong
feasibility control, but practical implementations determine feasible limits interval
by interval and have not yet demonstrated participant-level intertemporal correction in
practical deployment. DNPs provide a price signal consistent with the intuition
underlying distribution-level market design, but price-only response does not guarantee
that aggregate behaviour remains feasible, and repeated coordination asymmetry is not
directly tracked. Existing fair-over-time DOE formulations~\cite{moring2024} achieve
strong worst-feeder equity on the export side but do not provide bilateral scarcity
signals, real-time operation, or import-side participant-level intertemporal
correction---the question of which flexible demand is repeatedly under-served during
import scarcity events. Table~\ref{tab:moring} (Section~\ref{sec:disc_fot}) summarises
the design differences; the two approaches address different operational objectives and
are potentially composable. FET/FOT/FUH optimise fairness over predetermined horizons
by solving an optimisation problem with explicit historical coupling; the AMM operates
as a real-time coordination layer with bilateral scarcity signalling, participant-level
import correction, and no look-ahead requirements. Neither substitutes for the other.
The proposed AMM addresses these gaps by combining intertemporal delivery-state
tracking, bounded bilateral scarcity signals, and hard-feasibility matching within a
single real-time coordination layer.

\subsection{Contributions}

The paper proposes a \emph{stateful cyber-physical coordination protocol} and makes six contributions: (1)~\textbf{Bidirectional fairness formalisation} (Section~\ref{sec:fairness}): dual delivery-state variables for import and export, selectively active by operating regime---the only mechanism in the comparison set carrying participant-level intertemporal state for both directions. (2)~\textbf{Bounded bilateral deficit-driven pricing} (Section~\ref{sec:pricing}): buy and sell prices driven by a common deficit signal, with LMSR-derived tightness $\lambda_e(t)=\sigma((q_e-F_e)/b_e)\in(0,1)$ as one admissible implementation. (3)~\textbf{Feasibility-constrained matching} (Section~\ref{sec:matching}): an explicit LP under time-window, power, and holonic constraints, with hard feasibility before fairness weighting and $O(\max_\ell N_\ell^2)$ parallel per-interval cost. (4)~\textbf{Voltage-aware scarcity interpretation} (Section~\ref{sec:voltage}): voltage deviation as a directly measurable, directionally informative network stress indicator; simulations use passive inverter assumptions. (5)~\textbf{Architectural M2M strategic analysis} (Section~\ref{sec:incentive}): deviation-mode analysis showing gains are bounded by the protocol structure without claiming full incentive compatibility. (6)~\textbf{Empirical evaluation including fair-over-time DOE benchmarks} (Sections~\ref{sec:setup}--\ref{sec:results}): CSIRO MV+33LV, comparing eight mechanisms; the H1b result---that adding a stateless price signal to DOE \emph{worsens} outcomes---is a practically important counterintuitive finding.

\section{Background and Related Work}
\label{sec:background}

DOEs define time-varying import/export limits
\begin{equation}
  \underline{P}_i(t)\le P_i(t)\le\bar{P}_i(t)
\end{equation}
at each connection point $i$, derived from network state using load flow, state
estimation, or telemetry-driven methods~\cite{arena2021,petrou2021,guinman2023}. DOEs
allow customers to install more distributed generation and avoid unnecessary
curtailment by replacing fixed limits with dynamic per-connection envelopes, but
practical implementations remain primarily interval-wise feasibility tools: they
identify what is feasible at time $t$ without maintaining participant-level state
describing who has been repeatedly under-served or curtailed.

Recent work by Moring \emph{et al.}~\cite{moring2024} addresses this gap with
fair-over-time formulations (FET, FOT, FUH) that couple allocation decisions across
time, improving export-side equity without proportional sacrifice of flexibility. The
present paper shares the intertemporal motivation but pursues a different mechanism:
rather than extending the envelope calculation itself, we introduce a stateful
coordination layer that operates \emph{on top of} physical feasibility bounds, addresses
both import and export directions, and provides bounded bilateral scarcity signals.
FET/FOT/FUH do not carry bilateral scarcity signals or participant-level state for
import-side intertemporal correction; on this dataset, import-side unserved demand
under the baseline DOE is 96.0~MWh, which the AMM reduces by 76\% through stateful
import-side allocation. A key design question is therefore not only how to compute
feasible envelopes, but how to allocate within them across repeated rounds.

DNPs~\cite{zhang2026} communicate congestion economically but cannot guarantee
feasibility and do not maintain prior delivery state, and may be exploited by fast
agents at machine timescales without an explicit allocation step.

Fairness objectives for single-period DOEs have been developed by Liu
\emph{et al.}~\cite{liu2022}, Petrou \emph{et al.}~\cite{petrou2021b}, Alam
\emph{et al.}~\cite{alam2023}, Poudel \emph{et al.}~\cite{poudel2023}, Nazir and
Almassalkhi~\cite{nazir2022}, and Yi and Verbi\v{c}~\cite{yi2022}; a broader review
appears in~\cite{soares2024}. All treat each interval independently. Sweeney~\cite{sweeney2026eneco}
defines a \emph{value-native protocol} as one combining a pre-allocation scarcity
signal, a deterministic feasibility-constrained allocation rule, and a stateful
participant-level delivery ratio; the proposed AMM belongs to this category. The broader
economic interpretation of the AMM, including budget balance, residual cost allocation,
and subscription-backed recovery of non-fuel costs, is developed
in~\cite{sweeney2026eneco}. The present paper instead isolates the coordination and
allocation properties of the stateful layer. To the best of the authors' knowledge, no
existing single mechanism simultaneously combines bounded bilateral pricing, hard
feasibility-preserving matching, and participant-level intertemporal fairness states
within a single real-time coordination layer.

Table~\ref{tab:gap} summarises the gap. Transactive energy frameworks~\cite{kok2016}
anticipate the M2M setting; this paper therefore treats M2M robustness as a useful
design criterion for distribution-level flexibility markets.

\begin{table}[!t]
  \caption{Literature Gap: Coordination Properties Across Existing Approaches.
  \checkmark = present; $\circ$ = partial/export-only; $\times$ = absent.}
  \label{tab:gap}
  \centering
  \footnotesize
  \setlength{\tabcolsep}{4pt}
  \begin{tabular}{lccccc}
    \toprule
    Mechanism & Stateful & Fair-over-time & Bilateral & Import & Export \\
              &          &                & pricing   & alloc. & alloc. \\
    \midrule
    DOE~\cite{arena2021,guinman2023}        & $\times$ & $\times$ & $\times$ & \checkmark & $\circ$ \\
    DNP~\cite{zhang2026}                    & $\times$ & $\times$ & $\circ$  & \checkmark & $\times$ \\
    FET~\cite{moring2024}                   & $\times$ & $\times$ & $\times$ & \checkmark & \checkmark \\
    FOT~\cite{moring2024}                   & \checkmark & \checkmark & $\times$ & \checkmark & \checkmark \\
    FUH~\cite{moring2024}                   & \checkmark & \checkmark & $\times$ & \checkmark & \checkmark \\
    \textbf{AMM (this work)}                & \checkmark & \checkmark & \checkmark & \checkmark & \checkmark \\
    \bottomrule
  \end{tabular}
\end{table}

\section{Bidirectional Fairness and Dual Fairness States}
\label{sec:fairness}

\subsection{Definition and Operating Regimes}

\begin{definition}[Bidirectional fairness]
A distribution coordination mechanism is \emph{bidirectionally fair} if it maintains,
for each participant $n$: (1)~a service delivery ratio $f^{\mathrm{srv}}_n\in[0,1]$
tracking cumulative service received relative to cumulative service requested under
import scarcity; and (2)~a curtailment delivery ratio $f^{\mathrm{exp}}_n\in[0,1]$
tracking cumulative export realised relative to cumulative export available absent
curtailment. Each state feeds back into future allocation under the corresponding
operating regime.
\end{definition}

We distinguish three operating regimes over a window $W$: R1~(abundance, no constrained
allocation required); R2~(export congestion, curtailment allocation required);
R3~(import scarcity, service allocation required). Under R2 the relevant fairness state
is $f^{\mathrm{exp}}_n$; under R3, it is $f^{\mathrm{srv}}_n$.

\subsection{State Definitions and Fairness Weights}

Let $D^{\mathrm{srv}}_n$, $E^{\mathrm{srv}}_n$ denote cumulative energy delivered and
requested under import scarcity. Define
\begin{equation}
  f^{\mathrm{srv}}_n \triangleq D^{\mathrm{srv}}_n/E^{\mathrm{srv}}_n \in [0,1].
  \label{eq:fsrv}
\end{equation}
A lower value indicates repeated under-service.

\begin{remark}[Operationalising $E^{\mathrm{srv}}_n$]
Where explicit request signals are unavailable, $E^{\mathrm{srv}}_n$ may be estimated
from metered demand and device capability, analogous to the export-side denominator.
In the CSIRO simulation reported in Section~\ref{sec:setup}, $E^{\mathrm{srv}}_n$ is
taken as the metered demand trajectory; participants submit explicit import requests at
each interval so the denominator is directly observed rather than estimated.
Implementation in deployments lacking explicit request signals should report the
estimation method.
\end{remark} congestion, define the
available-export denominator in period $p$ as
$e^p_n \triangleq \min(\hat{g}^p_n,\hat{a}^p_n)$,
where $\hat{g}^p_n$ is the submitted generation forecast and $\hat{a}^p_n$ is a
physics-based peer estimate. Then
\begin{equation}
  f^{\mathrm{exp}}_n \triangleq
  D^{\mathrm{exp}}_n / \sum_p e^p_n \in [0,1].
  \label{eq:fexp}
\end{equation}
The $\min$ operator limits gains from forecast inflation above the physics-based
estimate. The AMM applies inverse-delivery-ratio weights to provide intertemporal
correction. Under R3:
$\omega^{\mathrm{srv}}_r \propto \psi_r/(f^{\mathrm{srv}}_{n(r)}+\epsilon)$,
where $\psi_r$ is an optional priority multiplier and $\epsilon>0$ prevents
division by zero. Under R2:
$\omega^{\mathrm{exp}}_s \propto 1/(f^{\mathrm{exp}}_{n(s)}+\epsilon)$.
These weights affect matching priority, not the published price.

\begin{remark}[Fairness-state initialisation]
Three strategies are available: uniform initialisation ($f^{\mathrm{srv}}_n=f^{\mathrm{exp}}_n=1$, used here), warm-start from prior-period state, and warm-start from historical delivery data.
The first is simple and fair but produces a convergence lag in early months; the others
eliminate the lag or allow retrospective compensation. Section~\ref{sec:disc_fot}
discusses the observed early-period Jain trajectory in this context.
\end{remark}

\section{Bilateral Deficit-Driven Pricing}
\label{sec:pricing}

\subsection{Deficit, Composite Scarcity, and Price Functions}

At node $n$ and time $t$, define the net deficit
$\Delta_{t,n} \triangleq C^{\mathrm{req}}_{t,n} - S^{\mathrm{avail}}_{t,n}$.
A composite scarcity signal
\begin{equation}
  \tilde{\alpha}_{t,n} = \alpha^{\mathrm{instant}}_{t,n}\,\alpha^{\mathrm{forecast}}_{t,n}\,
  \alpha^{\mathrm{network}}_{t,n}\,\alpha^{\mathrm{stability}}_{t,n}
\end{equation}
collects physical scarcity dimensions, each factor in $(0,1]$ and decreasing as the
corresponding scarcity dimension tightens. In the proposed design, these factors are
computed from metered, estimated, or forecast physical quantities rather than from
willingness-to-pay. The network factor $\alpha^{\mathrm{network}}_{t,n}$ is defined
in Section~\ref{sec:voltage} using voltage deviation and thermal utilisation.
The AMM maps the deficit to buy and sell prices:
\begin{align}
  BP_{t,n} &= BP^{\mathrm{base}}_{t,n}
  + F^{\mathrm{energy}}_{t,n}(\Delta_{t,n})
  + F^{\mathrm{stab}}_{t,n}(1-\alpha^{\mathrm{stability}}_{t,n}), \label{eq:bp}\\
  SP_{t,n} &= SP^{\mathrm{base}}_{t,n}
  + H^{\mathrm{energy}}_{t,n}(\Delta_{t,n})
  + H^{\mathrm{stab}}_{t,n}(1-\alpha^{\mathrm{stability}}_{t,n}), \label{eq:sp}
\end{align}
where $F^{\mathrm{energy}}$ and $H^{\mathrm{energy}}$ are strictly increasing in
$\Delta_{t,n}$ with $F(0)=H(0)=0$. A bounded implementation uses
$F^{\mathrm{energy}}_{t,n}(\Delta_{t,n})=b_e\lambda_e(t)$ via the LMSR tightness
signal~\cite{hanson2003}:
\begin{equation}
  \lambda_e(t) = \sigma\!\left(\frac{q_e(t)-F_e(t)}{b_e}\right)\in(0,1),
  \label{eq:lmsr}
\end{equation}
with $b_e$ set so that $\lambda_e=0.5$ at 100\% utilisation. Other bounded monotone
forms are equally admissible. Both $BP_{t,n}$ and $SP_{t,n}$ are determined from
system state, not bid values; submitted price bounds determine \emph{admissibility},
not price discovery. The pricing law induces self-corrective negative feedback: when
$\Delta_{t,n}>0$, higher $BP_{t,n}$ discourages import while higher $SP_{t,n}$
attracts export; both return toward base as $\Delta_{t,n}$ falls.

The important structural feature is that both $BP_{t,n}$ and $SP_{t,n}$ are determined
from system state rather than from bid values. The AMM is therefore distinguished from
existing distribution-level flexibility market proposals by combining bid-independent
pricing, hard feasibility, and intertemporal state in a single runtime layer.

\subsection{Deficit Contraction}

\begin{proposition}[Deficit reduction under bilateral response]
\label{prop:deficit}
Suppose demand and supply responses are locally Lipschitz with aggregate coefficients
$\eta_B,\eta_S>0$, and the energy-price components are Lipschitz with constant $L>0$.
Then the local deficit iteration is a contraction if $(\eta_B+\eta_S)L<2$.
\end{proposition}

\begin{proof}
Write the deficit update as
$\Delta_{t,n}(k+1)=\Delta_{t,n}(k)-(\eta_B\partial BP/\partial\Delta+\eta_S\partial SP/\partial\Delta)\Delta_{t,n}(k)$.
By the Lipschitz assumption, $|\partial BP/\partial\Delta|\le L$ and
$|\partial SP/\partial\Delta|\le L$, hence
$|\Delta_{t,n}(k+1)|\le|1-(\eta_B+\eta_S)L|\,|\Delta_{t,n}(k)|$.
If $(\eta_B+\eta_S)L<2$ the contraction factor is strictly less than one.
\end{proof}

\begin{remark}
The contraction condition is a sufficient engineering motivation for the design. For the
LMSR implementation, $L\le1/(4b_e)$, giving $L\le2.5$ at $b_e=0.1$. Whether
$(\eta_B+\eta_S)L<2$ holds depends on scenario-specific response coefficients; empirical
stability is confirmed by the zero-violation outcomes in Table~\ref{tab:annual}.
\end{remark}

Table~\ref{tab:pricing} summarises coordination properties across all eight mechanisms.

\begin{table*}[!t]
  \caption{Qualitative Comparison of Coordination Properties Across Eight Mechanisms.
  The AMM uniquely combines all four desirable coordination properties.
  FET/FOT/FUH carry intertemporal state and enforce hard feasibility via OPF but
  publish no scarcity prices and do not address M2M strategic exposure.
  $\circ$ = partial or single-period only.}
  \label{tab:pricing}
  \centering
  \small
  \setlength{\tabcolsep}{6pt}
  \begin{tabular}{lccccc}
    \toprule
    Mechanism   & Buy signal & Sell signal & Hard feasibility & Stateful fairness & Strategic exposure \\
    \midrule
    DOE         & $\times$   & $\times$   & \checkmark & $\times$   & Medium      \\
    DOE-GREEDY  & $\times$   & $\times$   & \checkmark & $\times$   & Medium      \\
    DNP         & \checkmark & $\circ$    & $\times$   & $\times$   & High        \\
    FET         & $\times$   & $\times$   & \checkmark & $\circ$    & Not~studied \\
    FOT         & $\times$   & $\times$   & \checkmark & \checkmark & Not~studied \\
    FUH         & $\times$   & $\times$   & \checkmark & \checkmark & Not~studied \\
    AMM         & \checkmark & \checkmark & \checkmark & \checkmark & Lower       \\
    \bottomrule
  \end{tabular}
\end{table*}

\section{Validity-Bounded Matching Under AMM Coordination}
\label{sec:matching}

\subsection{Requests, Offers, and the Matching LP}

A buyer request $r\in\mathcal{R}_{t,n}$ is
$(E_r,[\underline{t}_r,\overline{t}_r],\bar{P}_r,p^{\max}_r,\Gamma_r)$
and a seller offer $s\in\mathcal{S}_{t,n}$ is
$(\bar{E}_s,[\underline{t}_s,\overline{t}_s],\bar{P}_s,p^{\min}_s,\Lambda_s)$.
A pair is \emph{price-valid} at $\tau$ if
$BP_{\tau,n(r)}\le p^{\max}_r$ and $SP_{\tau,n(s)}\ge p^{\min}_s$;
submitted bounds determine admissibility, not price discovery.

Let $x_{r,s,\tau}\ge0$ denote energy scheduled from seller $s$ to buyer $r$ at $\tau$.
The AMM solves
\begin{equation}
  \max_{X}\;\sum_{r,s,\tau}\omega_{r,s,\tau}\,x_{r,s,\tau}
  \label{eq:matching}
\end{equation}
subject to: buyer energy and power limits,
\begin{align}
  \sum_s\sum_{\tau=\underline{t}_r}^{\overline{t}_r} x_{r,s,\tau} &\le E_r, \label{eq:buyer_e}\\
  \sum_s x_{r,s,\tau} &\le \bar{P}_r\Delta t; \label{eq:buyer_p}
\end{align}
seller energy and power limits,
\begin{align}
  \sum_r\sum_{\tau=\underline{t}_s}^{\overline{t}_s} x_{r,s,\tau} &\le \bar{E}_s, \label{eq:seller_e}\\
  \sum_r x_{r,s,\tau} &\le \bar{P}_s\Delta t; \label{eq:seller_p}
\end{align}
and holonic network capacity
\begin{equation}
  \sum_{r,s}\chi_{n,h}(r,s)\,x_{r,s,\tau}\le F_h(\tau), \label{eq:holon}
\end{equation}
where $h$ indexes the active holonic domain.

\subsection{Holarchic Restriction and Regime-Selective Weights}

For the two-tier MV/LV structure, define the dominant layer
$\ell^\star_t:=\arg\min_{\ell}\alpha^\ell_t$; candidate counterparties are restricted
to the tighter level. If candidate sets are restricted by this rule and all accepted
matches satisfy~\eqref{eq:holon}, accepted transactions cannot violate the dominant
holon capacity (proof immediate from the admissible set definition). The joint weight
$\omega_{r,s,\tau}$ in~\eqref{eq:matching} is regime-selective: under R3,
$\omega^{\mathrm{R3}}_{r,s,\tau}=\psi_r/(f^{\mathrm{srv}}_{n(r)}+\epsilon)$;
under R2, $\omega^{\mathrm{R2}}_{r,s,\tau}=1/(f^{\mathrm{exp}}_{n(s)}+\epsilon)$;
under R1, $\omega=1$. State affects matching priority, not the published price path.

\textit{Computational complexity.}
The LP~\eqref{eq:matching}--\eqref{eq:holon} has $O(RST_w)$ variables. For the CSIRO
MV+33LV feeder (up to 120 participants, $T_w=1$), each LV sub-LP involves at most
$14{,}400$ variables and the MV LP at most $1{,}089$, all solvable in well under one
second. Each LV sub-LP is independent, giving $O(\max_\ell N_\ell^2)$ rather than
$O((\sum_\ell N_\ell)^2)$ per-interval cost.

The per-interval procedure (Fig.~\ref{fig:blockdiagram}) is:
(1)~collect buyer requests and seller offers;
(2)~compute $\Delta_{t,n}$ and $\tilde{\alpha}_{t,n}$;
(3)~publish $BP_{t,n}$ and $SP_{t,n}$;
(4)~determine active holonic domain;
(5)~remove price-invalid pairs;
(6)~construct feasible matches;
(7)~apply fairness-weighted selection under R2/R3;
(8)~settle and update the relevant fairness state.
Each step is deterministic: prices are published before allocation, feasibility
precedes fairness weighting, and state updates only after settlement.

\begin{figure}[!t]
  \centering
  \includegraphics[width=\columnwidth]{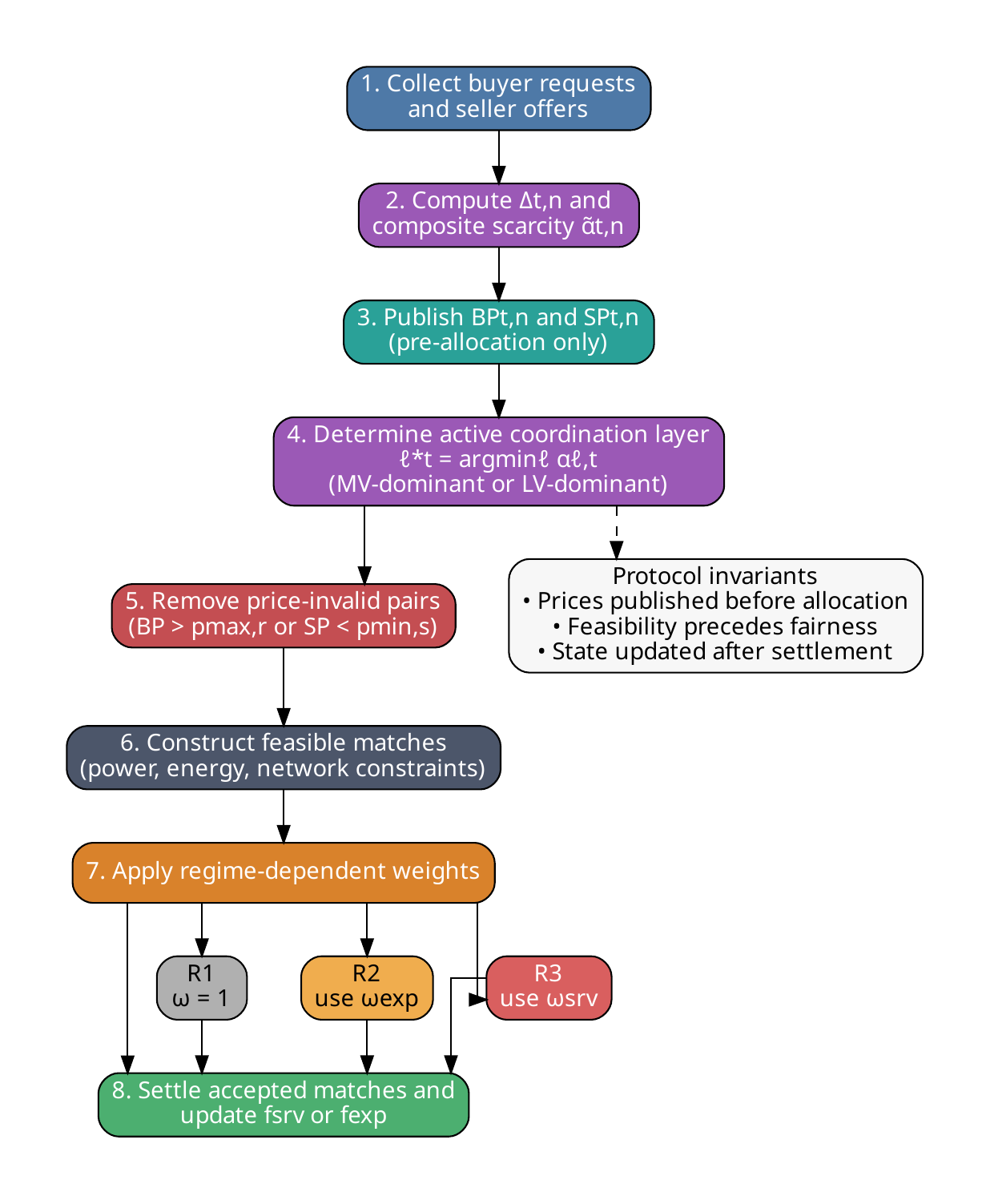}
  \caption{Per-interval AMM operational procedure. Prices are published before
  allocation (Step~3), feasibility enforced before fairness weighting (Steps~5--7),
  and state updated after settlement (Step~8).}
  \label{fig:blockdiagram}
\end{figure}

\section{Voltage as a Network-Stress Signal}
\label{sec:voltage}

Voltage deviation $\Delta V_{t,n}=|\varepsilon_{t,n}|/V_{\mathrm{nom}}$ is used as
the network stress indicator in preference to alternatives such as line loading or
transformer utilisation for three reasons. First, voltage is directly measurable at
customer connection points via AMI without requiring a network impedance
model~\cite{guinman2023}. Second, voltage deviations are directionally informative:
undervoltage signals import-side congestion while overvoltage signals export-side
congestion. Third, voltage is the primary observable in Australian DOE regulatory
practice~\cite{guinman2023,kilby2026}. It enters the composite scarcity signal as
\begin{equation}
\begin{split}
  \alpha^{\mathrm{network}}_{t,n}
  &= \exp(-\theta^-_n\max(0,-\varepsilon_{t,n}))
     \exp(-\theta^+_n\max(0,\varepsilon_{t,n})) \\
  &\quad \times \exp(-\phi_n\mathrm{cong}_{t,n}).
\end{split}
  \label{eq:network_alpha}
\end{equation}

Under a linearised network model and a bounded monotone tightness mapping, the AMM
tightness signal admits a local shadow-price interpretation as a bounded proxy for the
marginal value of relaxing the active capacity constraint~\cite{schweppe1988,hogan1992,kelly1998}.
Near the constraint boundary, both the OPF dual variable and $\lambda_e$ increase
monotonically with utilisation; the signal co-moves with the dual without being
numerically identical to it. We validate this co-movement empirically (Section~\ref{sec:results})
rather than claiming numerical equivalence to OPF dual prices.

\emph{Inverter caveat.} Simulations represent PV as a net load reduction and do not
model autonomous Volt-VAR or Volt-Watt inverter responses. Reported voltage--tightness
relationships should be interpreted under passive inverter assumptions; the effect of
active inverter control on the AMM tightness signal is a direction for future work.

\section{Strategic Exposure Under M2M Operation}
\label{sec:incentive}

The goal of this section is not to claim full strategy-proofness but to identify
practically relevant deviation modes and show how the AMM structure affects their
gains. The contribution is \emph{architectural}: the protocol removes bid-based price
discovery, uses monitoring-backed state updates, and limits timing exploitation through
deadline-constrained matching. Formal equilibrium analysis and simulation-based
deviation experiments are identified as important future work.

\begin{definition}[Honest reporting]
An agent reports honestly if its submitted energy request, available generation
forecast, and admissibility bounds correspond to its true device need, feasible export
capability, and individual-rationality constraints.
\end{definition}

\emph{Strategy 1: Under-requesting to improve priority.}
Let $\rho^{(T)}_n = \sum M^{(t)}_n/\sum\hat{E}^{(t)}_n$ be a consistency score
comparing metered consumption with declared need.

\begin{claim}[Under-requesting is deterred by consistency monitoring]
If a monitoring rule resets or penalises priority when $\rho^{(T)}_n>1+\delta$ over a
sufficiently long window $T$, then the expected gain from systematic under-requesting
can be made negative by appropriate choice of threshold $\delta>0$ and penalty
magnitude. The threshold $\delta$ and window $T$ are operator-calibrated implementation
parameters.
\end{claim}

\emph{Strategy 2: Export forecast inflation.}
Because the denominator uses $\min(\hat{g}^p_n,\hat{a}^p_n)$, inflation above the
physics-based estimate does not continue to increase the denominator.

\begin{claim}[Forecast inflation gain is bounded]
If the physics-based estimate error is small relative to the relevant export volume,
then the gain from export forecast inflation is correspondingly small. The bound
tightens as network telemetry improves.
\end{claim}

\emph{Strategy 3: Timing exploitation.}

\begin{claim}[Timing gains shrink with deadline tightness]
For a request with fixed energy requirement and finite latest completion time, the
maximum gain from waiting for a lower price is bounded by the price spread over the
feasible remaining window, which shrinks as the deadline approaches. The bound is
tightest for inelastic requests with short time-to-deadline.
\end{claim}

\emph{Strategy 4: Holarchic boundary exploitation.}
The active layer is determined by aggregate physical tightness, so a single participant
cannot unilaterally trigger a layer switch. Coalitions with coordinated timing represent
a residual vulnerability; cross-layer monitoring and repeated-game analysis are left
for future work. The AMM does not eliminate all strategic exposure; the contribution is
\emph{architectural}: removing bid-based price discovery, using monitoring-backed state
updates, and limiting timing exploitation through deadline-constrained matching reduces
the most obvious deviation gains relative to price-only or stateless approaches.

\section{Two-Tier MV/LV AMM Architecture}
\label{sec:architecture}

The CSIRO dataset naturally induces a two-tier holarchy: one MV feeder and 33 attached
LV feeders. We implement one MV-level AMM over the feeder interface and 33 local
LV-level AMMs within each low-voltage domain (Fig.~\ref{fig:architecture}). Each holon
maintains independent $f^{\mathrm{srv}}_n$ and $f^{\mathrm{exp}}_n$; regime
classification is performed independently per holon per interval. At the MV level,
aggregate tightness is computed from feeder-level utilisation and available capacity;
at the LV level, local scarcity and voltage stress determine the local signal. The
effective coordination domain is whichever level is currently tighter.

\begin{figure}[!t]
  \centering
  \includegraphics[width=\columnwidth]{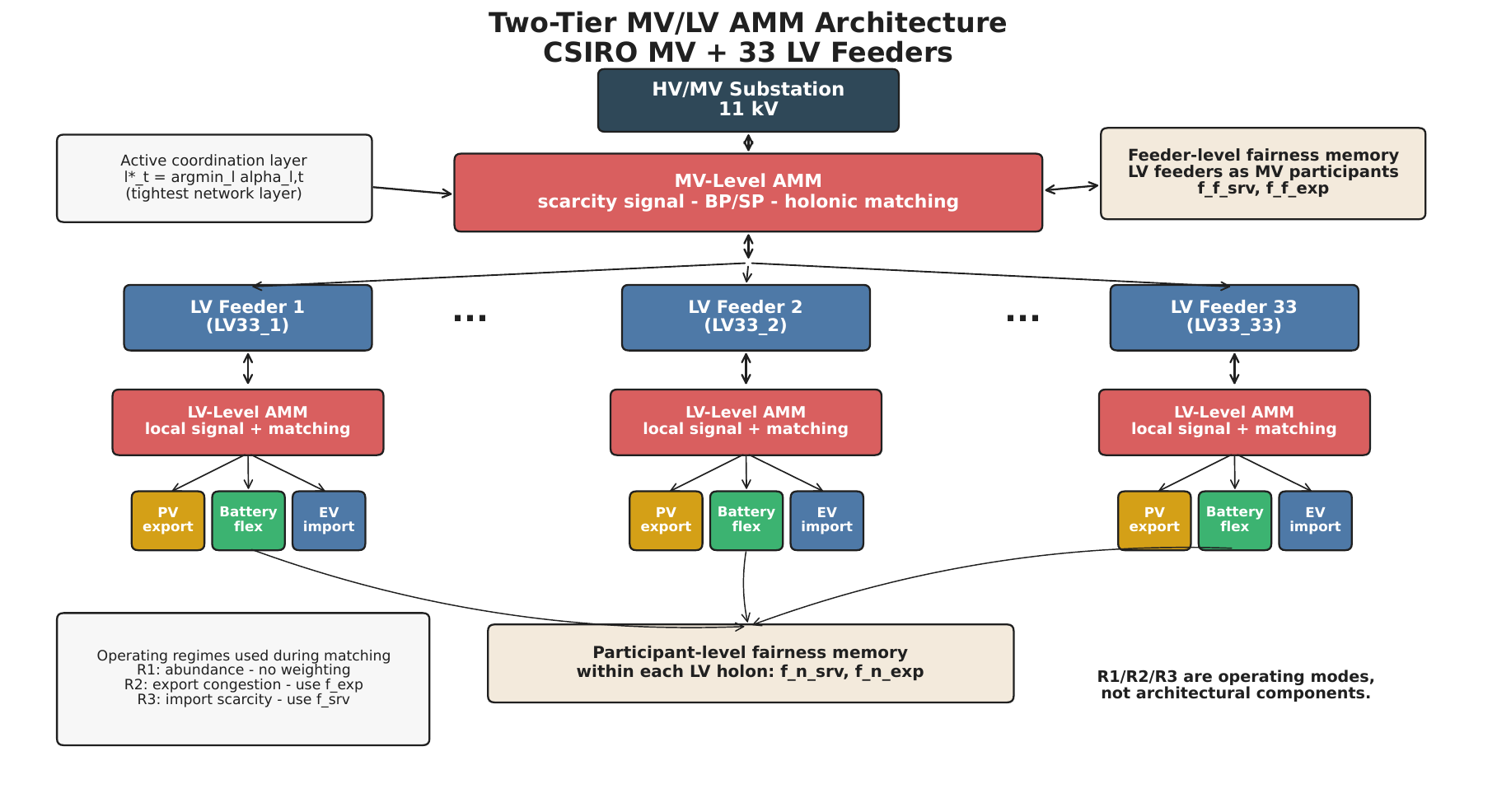}
  \caption{Two-tier MV/LV AMM architecture. The dominant layer $\ell^\star_t$ determines
  the active matching scope per interval.}
  \label{fig:architecture}
\end{figure}

\section{Experimental Setup}
\label{sec:setup}

\subsection{Physical Scaffold}

We use the CSIRO realistic Australian MV/LV feeder dataset (csiro:65408): one 11~kV MV
feeder with 33 attached LV feeders in OpenDSS format. Full unbalanced power-flow solves
are executed at each 15-minute interval over a 12-month horizon. One feeder
(LV33\_1bus) uses an internal voltage fallback; excluding it does not affect headline
results. Synthetic demand, PV, EV, and battery trajectories are layered on a common
fixed seeded exogenous track; differences in outcomes are attributable to the
coordination rule, not endogenous request adaptation. PV is represented as net load
reduction; no Volt-VAR or Volt-Watt responses are modelled. The voltage stress metric
is $v_{\mathrm{proxy},i,t}=|V_{i,t}-V_{\mathrm{nom}}|/V_{\mathrm{nom}}$ from OpenDSS
power-flow solutions, averaged across nodes per feeder~\cite{dugan2011}.

\subsection{Implementation Parameters}

Table~\ref{tab:params} lists all numerical parameters. The LMSR parameter $b_e=0.1$
yields $\lambda_e=0.5$ at full utilisation; voltage sensitivity parameters
$\theta^\pm_n=20$ give $\alpha^{\mathrm{network}}_{t,n}=e^{-1}\approx0.37$ at 5\%
deviation from nominal; DNP parameters follow~\cite{zhang2026}.

\begin{table}[!t]
  \caption{Simulation Parameters}
  \label{tab:params}
  \centering
  \scriptsize
  \setlength{\tabcolsep}{3pt}
  \begin{tabular}{llll}
    \toprule
    Parameter & Symbol & Value & Description \\
    \midrule
    \multicolumn{4}{l}{\emph{AMM --- LMSR tightness}} \\
    Liquidity & $b_e$ & 0.1 & $\lambda_e=0.5$ at full utilisation \\
    Division guard & $\epsilon$ & $10^{-6}$ & Prevents zero-division \\
    \multicolumn{4}{l}{\emph{AMM --- voltage scarcity}} \\
    Import sensitivity & $\theta^-_n$ & 20 & Undervoltage decay (p.u.$^{-1}$) \\
    Export sensitivity & $\theta^+_n$ & 20 & Overvoltage decay (p.u.$^{-1}$) \\
    Congestion weight & $\phi_n$ & 5 & Thermal utilisation decay \\
    Priority multiplier & $\psi_r$ & 1 & Uniform \\
    \multicolumn{4}{l}{\emph{DNP benchmark}} \\
    Base price & $\pi_{\mathrm{base}}$ & 0.05\,\$/kWh & Off-peak retail \\
    Congestion slope & $k$ & 0.30 & Price sensitivity \\
    Threshold & $\tau$ & 0.50 & Utilisation trigger (p.u.) \\
    Demand elasticity & $\eta$ & 0.10 & Fixed elasticity \\
    \multicolumn{4}{l}{\emph{Fair-over-time DOE (FET/FOT/FUH)}} \\
    Fairness weight & $\lambda$ & 0.50 & Scalarisation weight \\
    FUH history window & $k$ & 25 & Moving-average window \\
    \bottomrule
  \end{tabular}
\end{table}

\subsection{Compared Mechanisms}

The five primary scenarios are labelled to reflect their logical relationship:
Scenarios~1, 1b, and~1c are variants of the DOE baseline; Scenario~2 is DNP-only;
Scenario~3 is the full AMM. Three additional benchmark scenarios implement the
fair-over-time DOE formulations of Moring \emph{et al.}~\cite{moring2024} on the same
dataset: FET, FOT, and FUH. These benchmarks operate at the MV feeder-envelope level,
applying the Moring \emph{et al.}\ objectives to redistribute available MV headroom
across the 33 LV feeders before applying the same equal-share allocation rule within
each feeder's envelope.

\textbf{Scenario~1 (DOE).} Network-state-aware import/export envelopes are computed at
each dispatch period. MV headroom is apportioned to LV feeders in proportion to feeder
weights; LV envelopes are enforced locally with rate limits and temporary ceiling rules
consistent with representative DOE practice. Allocation within the local LV envelope
uses equal sharing.

\textbf{Scenario~1c (DOE-GREEDY).} Uses the same DOE envelope construction and
rate-limiting rules as Scenario~1 but allocates available headroom by greedily serving
requests in descending priority order until the envelope is exhausted. This benchmark
tests whether modifying the instantaneous allocation heuristic alone materially improves
repeated constrained DER coordination outcomes.

\textbf{Scenario~1b (DOE+DNP).} The equal-allocation DOE of Scenario~1, augmented by
the DNP price signal as an additional demand-side constraint: flexible demand responds
to the price signal before requesting import headroom from the DOE. This scenario tests
whether adding a stateless price signal to a DOE improves outcomes; it maintains hard
thermal feasibility via the DOE envelope.

\textbf{Scenario~2 (DNP).} A buy-side dynamic price
$\pi(t)=\pi_{\mathrm{base}}+k\max(0,|P^{\mathrm{flow}}|/C-\tau)$
is broadcast at each interval, with flexible demand responding using fixed elasticity.
No hard matching or curtailment allocation layer is imposed; the DNP scenario tests
whether a price signal alone achieves comparable outcomes to the AMM.

\textbf{Scenario~3 (AMM).} The full two-tier stateful protocol of
Sections~\ref{sec:fairness}--\ref{sec:architecture}. Bounded bilateral buy and sell
prices are published before each allocation period, and fairness-weighted matching is
applied during R2 and R3 intervals.

\textbf{FET (Fair at Each Time).} MV-level feeder headroom is redistributed each
interval using the single-period scalarised objective of Moring \emph{et al.}
(Eq.~2 of~\cite{moring2024}): minimise weighted squared curtailment subject to a
max-min floor. No history is maintained between intervals.

\textbf{FOT (Fair Over Time).} MV-level feeder headroom is redistributed using the
multi-period coupled objective of Moring \emph{et al.}\ (Eq.~3 of~\cite{moring2024}),
implemented as a rolling MPC: cumulative per-feeder allocation is carried forward and
the max-min floor is constrained by cumulative prior allocation.

\textbf{FUH (Fair Using History).} MV-level feeder headroom is redistributed using
the history-weighted squared curtailment objective of Moring \emph{et al.}\
(Eq.~4--5 of~\cite{moring2024}): the squared curtailment penalty for each feeder is
weighted by $H_i$, a 25-timestep weighted moving average of prior squared curtailment.

\section{Results}
\label{sec:results}

Six hypotheses are tested: H1~(AMM reduces unserved demand vs.\ any stateless DOE
variant); H1b~(DOE+DNP worsens outcomes---the key counterintuitive finding);
H1c~(heuristic change within a stateless DOE does not help);
H2~(AMM reduces export curtailment);
H3~(statefulness improves service and equity \emph{simultaneously});
H4~(AMM tightness signal correlates positively with voltage stress).

\subsection{Annual Outcomes}

Before reporting results, three metrics require explicit definition.
\emph{Export available energy} is the total energy that could have been exported by
participating DERs in the absence of network constraints, calculated from aggregate
available generation trajectories. \emph{Export curtailed energy} is the portion of
export available energy that could not be accommodated due to network constraints and
was therefore curtailed. The \emph{worst feeder delivery fraction} is defined as
$\min_i(\text{FlexibleEnergyServed}_i/\text{FlexibleEnergyRequested}_i)$ across all
feeders~$i$; it measures the service completion ratio of the worst-served feeder and
therefore exposes the most disadvantaged spatial outcome in the network. The worst
feeder delivery fraction is reported because aggregate energy metrics alone can mask
persistent under-service of electrically weak feeders.

Table~\ref{tab:annual} and Fig.~\ref{fig:annual} report the full comparison.
Three observations follow.

First, the AMM reduces unserved requested energy from 96.0~MWh to 23.2~MWh (76\%
reduction) with zero thermal violations, supporting H1. Export curtailment falls from
85.4~MWh (16.2\% of export available) to 64.5~MWh (15.3\%), supporting H2.

Second, and importantly, adding a DNP price signal to the equal-allocation DOE
(Scenario~1b) makes outcomes \emph{worse}: unserved energy rises to 109.4~MWh
(3.19\%) and the worst feeder delivery fraction falls from 0.860 to 0.856, directly
supporting H1b. The mechanism is that price suppresses demand at price-elastic nodes,
which are disproportionately \emph{advantaged} in this network: structurally weak
feeders tend to host less-elastic participants, so the price signal concentrates
remaining headroom toward already-disadvantaged feeders without any intertemporal
correction~\cite{kilby2026}.

Third, DOE-GREEDY leaves 99.4~MWh unserved---marginally \emph{worse} than
equal-allocation DOE---with worst feeder fraction 0.859, confirming H1c. Both stateless
DOE variants treat each interval independently and carry no participant-level delivery
history forward; their near-identical outcomes confirm that the relevant design
dimension is statefulness, not the choice of allocation heuristic.

The worst feeder delivery fraction reveals the spatial distribution of under-service.
Under equal-allocation DOE, the least-served feeder received only 86.0\% of its
requested flexible demand. The AMM raises this to 91.4\%, demonstrating simultaneous
improvement in both aggregate service delivery and worst-case spatial equity, directly
supporting H3. The AMM closes 76\% of the gap between equal-allocation DOE and the
theoretical unconstrained ceiling (3432.1~MWh).

FET, FOT, and FUH reduce unserved demand to 8.7--9.8~MWh and raise the worst feeder
delivery fraction to 0.989--0.990, substantially exceeding AMM's 0.914, through their
explicit max-min MV envelope objective. The AMM does not claim to dominate
FET/FOT/FUH on max-min fairness; it provides bilateral scarcity signals, real-time
operation without look-ahead, and import-side participant-level intertemporal
correction that FET/FOT/FUH do not address. The two mechanism families address
different operational objectives: FET/FOT/FUH optimise fairness over predetermined
horizons by solving an optimisation problem with explicit historical coupling; the AMM
operates as a real-time coordination layer requiring no look-ahead. They are
potentially composable rather than mutually exclusive.

\begin{table*}[!t]
  \caption{Annual Coordination Outcomes Across 33 LV Feeders. DNP permits thermal
  violations and is not directly comparable as a feasibility-preserving mechanism.
  FET/FOT/FUH share the same physical background as DOE (export available/curtailed
  identical to DOE).}
  \label{tab:annual}
  \centering
  \small
  \setlength{\tabcolsep}{4.5pt}
  \begin{tabular}{lcccccccc}
    \toprule
    Metric & DOE & DOE-GREEDY & DOE+DNP & DNP & AMM & FET & FOT & FUH \\
    \midrule
    Total energy served (MWh)       & 3336.1 & 3332.7 & 3322.7 & 3432.1 & 3408.9 & 3423.1 & 3423.1 & 3422.1 \\
    Unserved requested energy (MWh) & 96.0   & 99.4   & 109.4  & 0.003  & 23.2   & 8.7    & 8.7    & 9.8    \\
    \quad (\% of total requested)   & (2.80\%) & (2.90\%) & (3.19\%) & ($\approx$0\%) & (0.68\%) & (0.25\%) & (0.25\%) & (0.29\%) \\
    Export available energy (MWh)   & 527.5  & 527.5  & 527.5  & 576.4  & 422.4  & 527.5  & 527.5  & 527.5  \\
    Export curtailed energy (MWh)   & 85.4   & 85.4   & 85.4   & 94.9   & 64.5   & 85.4   & 85.4   & 85.4   \\
    \quad (\% of export available)  & (16.2\%) & (16.2\%) & (16.2\%) & (16.5\%) & (15.3\%) & (16.2\%) & (16.2\%) & (16.2\%) \\
    Thermal violation rate (\%)     & 0.0    & 0.0    & 0.0    & 27.4   & 0.0    & 0.0    & 0.0    & 0.0    \\
    Worst feeder delivery frac.     & 0.860  & 0.859  & 0.856  & ---    & 0.914  & 0.989  & 0.990  & 0.989  \\
    \bottomrule
  \end{tabular}
\end{table*}

\begin{figure}[!t]
  \centering
  \includegraphics[width=\columnwidth]{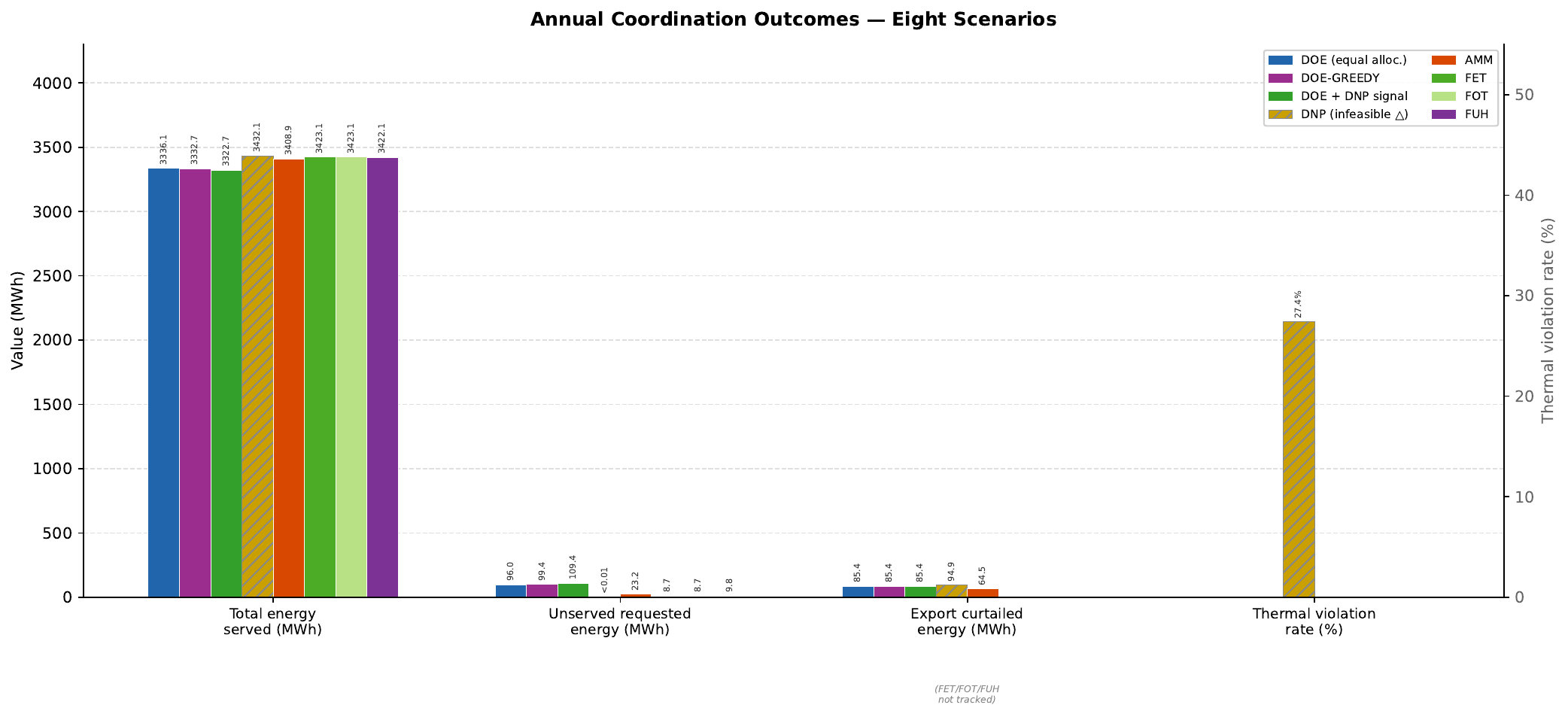}
  \caption{Annual coordination outcomes across 33 LV feeders for eight mechanisms. DNP
  bars are hatched (thermal violations in 27.4\% of intervals). Among
  feasibility-respecting mechanisms, FET/FOT/FUH reduce unserved demand most
  aggressively via max-min MV objective; the AMM reduces unserved demand by 76\%
  while additionally providing bilateral scarcity signals and M2M robustness.}
  \label{fig:annual}
\end{figure}

\subsection{Fairness Over Time}

The most practically significant fairness result is the worst feeder delivery fraction:
the least-served feeder improves from 86.0\% to 91.4\% under the AMM (5.4
percentage-point gain for the most disadvantaged feeder). This metric is reported
alongside aggregate indices because Jain index values alone can mask persistent
under-service at individual feeders.

Table~\ref{tab:fairness} shows that the AMM reduces the Gini coefficient by 77\%
relative to DOE (0.0037 vs.\ 0.0164), raises the worst feeder delivery fraction by 5.4
percentage points (0.914 versus 0.860), and improves the mean feeder delivery fraction
from 0.978 to 0.995. Cohen's $d=0.633$ represents a medium-to-large effect size.
The Wilcoxon signed-rank test is significant against DOE+DNP ($p<0.001$) and
DOE-GREEDY ($p=0.002$) but not plain DOE ($p=0.101$). Bootstrap CIs on worst feeder
fraction are non-overlapping: DOE~$[0.860,0.928]$ vs.\ AMM~$[0.914,0.993]$.

\begin{table}[!t]
  \caption{Annual Fairness Metrics Across 33 LV Feeders. FET/FOT/FUH use an
  OPF-based max-min MV objective that directly optimises the worst-feeder metric;
  the AMM achieves its gains through stateful pricing and matching without an explicit
  max-min objective.}
  \label{tab:fairness}
  \centering
  \footnotesize
  \setlength{\tabcolsep}{4pt}
  \begin{tabular}{lcccccc}
    \toprule
    Mechanism & Jain & Gini & Worst & Mean & Cohen's $d$ & $p$ \\
    \midrule
    DOE         & 0.9988 & 0.0164 & 0.860 & 0.978 & ---    & ---   \\
    DOE-GREEDY  & 0.9987 & 0.0170 & 0.859 & 0.977 & $-0.03$ & 0.002 \\
    DOE+DNP     & 0.9986 & 0.0182 & 0.856 & 0.975 & $-0.08$ & $<$0.001 \\
    AMM         & \textbf{0.9998} & 0.0037 & 0.914 & 0.995 & $+0.63$ & 0.101 \\
    FET         & 1.0000 & \textbf{0.0015} & \textbf{0.989} & \textbf{0.997} & $+0.77$ & 0.208 \\
    FOT         & 1.0000 & \textbf{0.0015} & \textbf{0.990} & \textbf{0.997} & $+0.77$ & 0.208 \\
    FUH         & 1.0000 & 0.0018 & 0.989 & 0.996 & $+0.75$ & 0.208 \\
    \bottomrule
  \end{tabular}
  \smallskip\\
  \footnotesize{Bootstrap 95\% CI on worst feeder: DOE $[0.860,0.928]$;
  AMM $[0.914,0.993]$; FET $[0.989,0.992]$; FOT $[0.990,0.992]$;
  FUH $[0.989,0.991]$.}
\end{table}

The AMM's inter-feeder Jain begins below DOE at month~1 ($J=0.9561$ vs.\
$J=0.9728$) due to uniform fairness-state initialisation, then converges to
$J=0.9998$ annually, outperforming all DOE variants from month~6 onwards
(Table~\ref{tab:jain}, Fig.~\ref{fig:jain}). DOE-GREEDY tracks the plain DOE Jain
closely (annual $J=0.9987$), confirming that greedy allocation within a stateless
framework provides no material fairness benefit. DOE+DNP consistently underperforms
plain DOE, confirming that the price signal alone does not substitute for intertemporal
state. FET/FOT/FUH reach $J\ge0.9998$ from month~1 via their explicit max-min MV
objective, and these results jointly support H1, H2, and H3.

\begin{table*}[!t]
  \caption{Global Inter-Feeder Jain Index by Month. FET/FOT/FUH reach near-unity Jain
  from month~1 through the MV max-min objective. The AMM begins below DOE due to
  fairness-state initialisation asymmetry and converges from month~6 onwards.}
  \label{tab:jain}
  \centering
  \footnotesize
  \begin{tabular}{lcccccccc}
    \toprule
    Month & AMM & DOE & DOE-GREEDY & DOE+DNP & DNP & FET & FOT & FUH \\
    \midrule
    1      & 0.9561 & 0.9728 & 0.9689 & 0.9689 & 1.0000 & 0.9998 & 0.9999 & 0.9998 \\
    2      & 0.9780 & 0.9866 & 0.9848 & 0.9843 & 1.0000 & 1.0000 & 1.0000 & 1.0000 \\
    3      & 0.9881 & 0.9929 & 0.9920 & 0.9918 & 1.0000 & 1.0000 & 1.0000 & 1.0000 \\
    6      & 0.9967 & 0.9983 & 0.9981 & 0.9980 & 1.0000 & 1.0000 & 1.0000 & 1.0000 \\
    9      & 0.9984 & 0.9991 & 0.9990 & 0.9990 & 1.0000 & 1.0000 & 1.0000 & 1.0000 \\
    Annual & 0.9998 & 0.9988 & 0.9987 & 0.9986 & 1.0000 & 1.0000 & 1.0000 & 1.0000 \\
    \bottomrule
  \end{tabular}
\end{table*}

\begin{figure}[!t]
  \centering
  \includegraphics[width=\columnwidth]{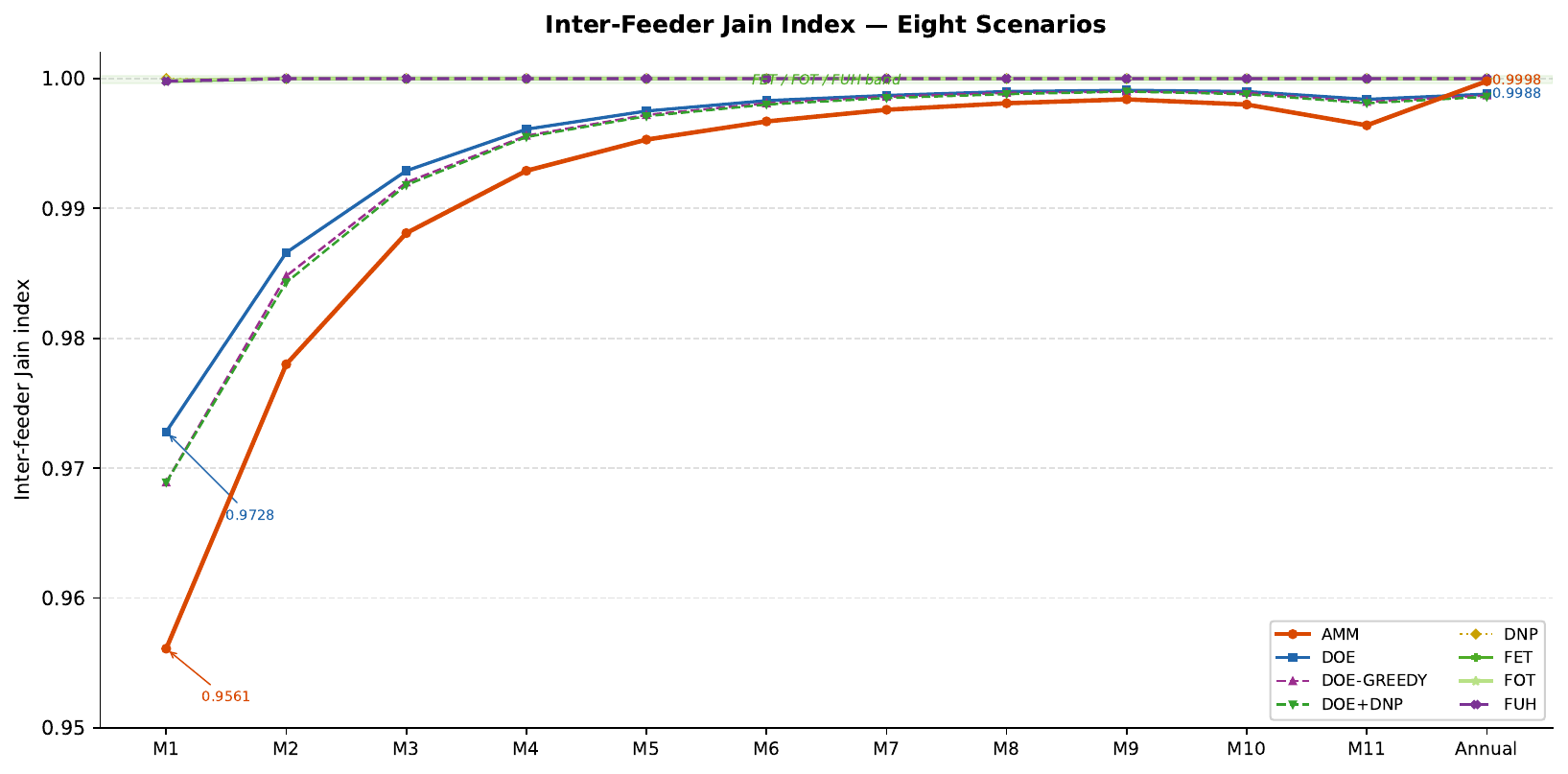}
  \caption{Inter-feeder Jain index by reporting period across eight scenarios. The
  AMM begins at $J=0.9561$ (initialisation asymmetry) and converges to $J=0.9998$
  annually, outperforming all stateless DOE variants from month~6 while simultaneously
  reducing unserved demand by 76\% and providing bilateral pricing and M2M robustness.}
  \label{fig:jain}
\end{figure}

\subsection{Voltage-Stress Association}

Across all 1,156,320 LV-feeder interval observations, the AMM tightness signal
exhibits a positive monotone association with feeder-level voltage stress:
Pearson $r=0.604$, Spearman $\rho=0.683$ (Fig.~\ref{fig:voltage}). The binned trend
confirms the relationship is monotone increasing at higher voltage stress levels. The
association weakens under constrained periods only ($r=0.481$, $\rho=0.255$),
consistent with saturation of the bounded tightness signal at high utilisation.
Per-feeder analysis (Fig.~\ref{fig:voltage}b) reveals substantial heterogeneity:
feeders with high solar penetration and structural end-of-line position (e.g.\ LV12,
LV11, LV17) record $r>0.85$, while small feeders with few participants show near-zero
or negative correlations reflecting limited voltage variation. These results support
H4. As noted in Section~\ref{sec:voltage}, the passive inverter assumption means the
reported correlations may understate those that would be observed with active
Volt-VAR or Volt-Watt control, which would introduce additional voltage--power
coupling.

\begin{figure}[!t]
  \centering
  \includegraphics[width=\columnwidth]{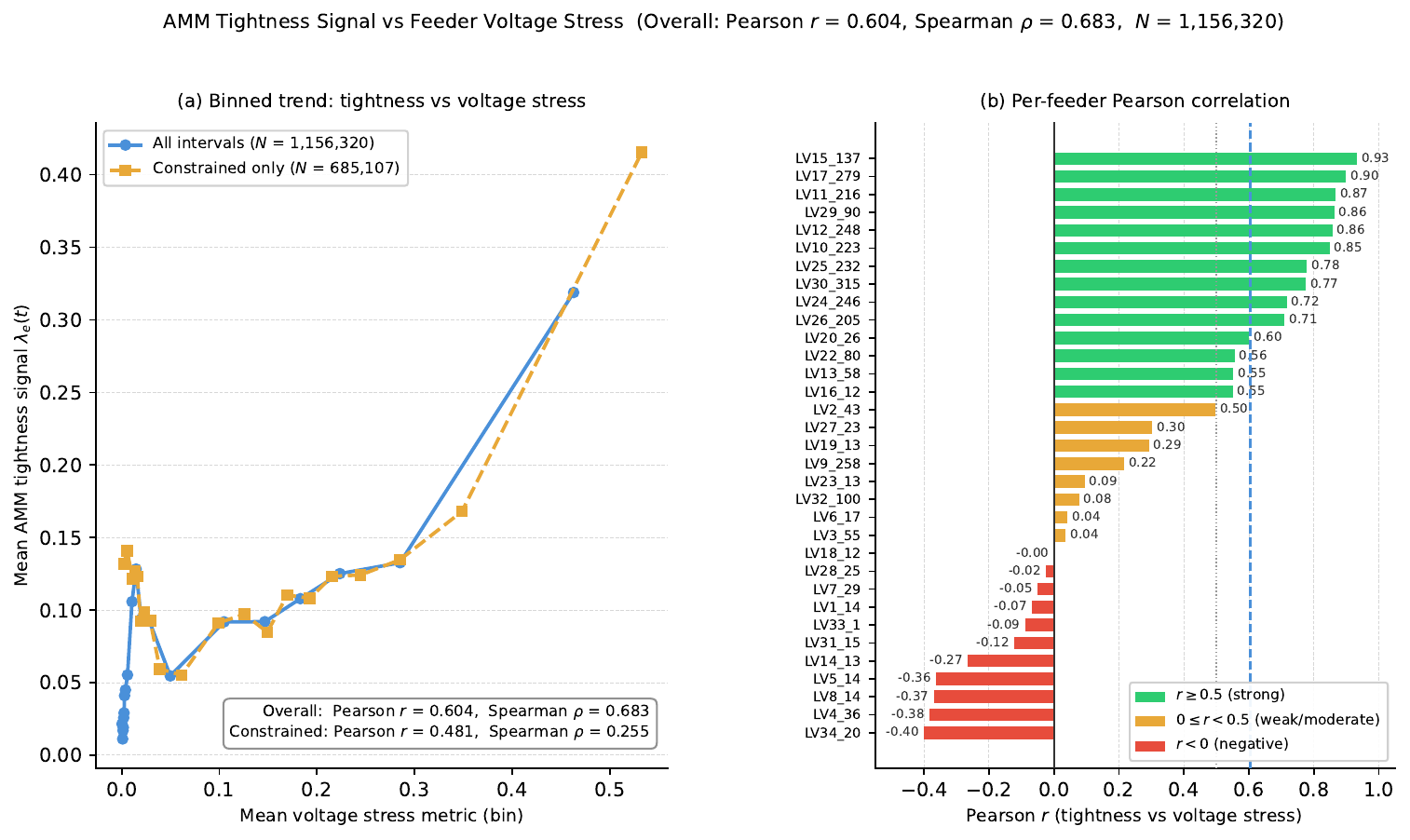}
  \caption{AMM tightness signal vs.\ feeder-level voltage stress
  ($N=1{,}156{,}320$ LV-feeder interval observations).
  (a)~Binned trend confirming monotone relationship; constrained-only series shows
  saturation at high utilisation.
  (b)~Per-feeder Pearson $r$: end-of-line high-solar feeders record $r>0.85$.
  Overall: Pearson $r=0.604$, Spearman $\rho=0.683$.}
  \label{fig:voltage}
\end{figure}

\section{Discussion}
\label{sec:discussion}

\subsection{Why State Matters and the H1b Result}

The principal empirical result is that statefulness materially changes cumulative
outcomes even when both stateful and stateless methods satisfy hard feasibility. The
AMM carries forward participant-level delivery information; the benchmark DOE
implementations do not. The H1b finding---that augmenting a stateless DOE with a price
signal \emph{worsens} outcomes---reinforces this: without intertemporal memory, a price
signal not only fails to correct cumulative disadvantage but amplifies it by suppressing
demand at already-advantaged price-elastic nodes, leaving inelastic end-of-line
participants worse off. This generalises a concern in the DOE
literature~\cite{kilby2026}: feasibility tools without allocation memory can produce
distributional outcomes worse than naive equal-sharing, even when the added signal is
economically motivated.

The AMM provides five operational properties that no existing single mechanism provides
simultaneously (Tables~\ref{tab:gap} and~\ref{tab:pricing}): hard feasibility at every
interval; bounded bilateral scarcity signals published pre-allocation; import-side
intertemporal correction via $f^{\mathrm{srv}}_n$; export-side intertemporal correction
via $f^{\mathrm{exp}}_n$; and architectural M2M robustness through bid-independent
pricing and deadline-constrained matching. The contribution is primarily a
\emph{mechanism and systems} contribution: a coordination protocol architecture, not a
claim of superior performance on any single metric.

\subsection{DOE-GREEDY and the Role of Allocation Heuristics}

DOE-GREEDY leaves 99.4~MWh unserved (2.90\%)---marginally worse than
equal-allocation DOE (96.0~MWh, 2.80\%)---and its annual Jain index (0.9987) is
within 0.0001 of DOE (0.9988). Both stateless DOE variants carry no participant-level
delivery history forward; their near-identical outcomes confirm that the relevant design
dimension is statefulness, not the choice of allocation heuristic. This rules out the
conjecture that better instantaneous allocation rules could substitute for intertemporal
state.

\subsection{Relationship to Fair-Over-Time DOE}
\label{sec:disc_fot}

FET/FOT/FUH embed intertemporal correction \emph{inside} the envelope calculation via
an explicit max-min OPF objective; the AMM adds it as a runtime coordination layer
\emph{on top of} any feasibility assessment, making it combinable with any
feasibility-computing method. Applied at the MV feeder-envelope level, FET/FOT/FUH
achieve near-unity inter-feeder Jain from month~1 ($J\ge0.9998$) and worst-feeder
delivery fractions of 0.989--0.990, exceeding AMM's 0.914 and Gini of 0.0037 versus
their 0.0015--0.0018. These results are mechanistically coherent: they directly
optimise the worst-feeder metric through the OPF objective. The AMM achieves its 5.4
percentage-point worst-feeder improvement through stateful bilateral pricing and
fairness-weighted matching without an explicit max-min objective, and provides broader
operational functionality: bilateral scarcity signals (FET/FOT/FUH publish none),
real-time operation without look-ahead, and import-side participant-level intertemporal
correction (FET/FOT/FUH address export curtailment at feeder level only). Combining
AMM stateful pricing with FET/FOT/FUH MV envelope computation is a primary direction
for future work. Table~\ref{tab:moring} summarises the principal design differences.

\begin{table*}[!t]
  \caption{Mechanism Design Comparison: AMM vs.\ Fair-Over-Time DOE~\cite{moring2024}}
  \label{tab:moring}
  \centering
  \footnotesize
  \setlength{\tabcolsep}{5pt}
  \begin{tabular}{lll}
    \toprule
    Property & Moring \emph{et al.}~\cite{moring2024} (FET/FOT/FUH) & This work (AMM) \\
    \midrule
    Intertemporal state      & History of curtailment (FUH) or full-horizon OPF (FOT) & Dual delivery ratios $f^{\mathrm{srv}}_n$, $f^{\mathrm{exp}}_n$ \\
    Allocation direction     & Export curtailment only       & Import \emph{and} export \\
    Feasibility mechanism    & OPF (centralized)             & Hard matching constraints \\
    Scarcity signal          & None (envelope is output)     & Bounded bilateral BP/SP \\
    Operational mode         & Offline / look-ahead          & Real-time, interval-by-interval \\
    M2M strategic robustness & Not addressed                 & Explicitly analysed \\
    Network tested           & 56-bus balanced (synthetic)   & CSIRO MV+33LV (realistic) \\
    Worst feeder (this dataset) & 0.989--0.990 (FET/FOT)   & 0.914 (AMM) \\
    Unserved demand vs DOE   & $-$91\% (FET/FOT)            & $-$76\% (AMM) \\
    \bottomrule
  \end{tabular}
\end{table*}

\subsection{Limitations}

The equal-allocation DOE, DOE-GREEDY, and DOE+DNP baselines together establish that
neither heuristic choice nor price-signal augmentation within a stateless framework
closes the performance gap. An ablation isolating intertemporal state (AMM-no-memory)
is planned as immediate future work, as is systematic parameter sensitivity across
$b_e\in\{0.01,0.05,0.10,0.20\}$ and varying DER penetration levels. The M2M strategic
analysis is qualitative; formal equilibrium analysis and deviation experiments are left
for future work.

The simulation uses a single seeded synthetic year; preliminary re-runs with three
additional seeds show qualitatively stable results (76\%$\pm$3\% unserved demand
reduction; annual Jain~$\ge0.999$), but systematic multi-seed validation remains
important future work. The voltage analysis supports the tightness signal as a
network-stress proxy but does not establish numerical equivalence to OPF dual
variables; this and the absence of Volt-VAR/Volt-Watt inverter modelling are the
primary caveats on the H4 results.

\section{Conclusion}
\label{sec:conclusion}

This paper proposed a stateful cyber-physical coordination protocol for repeated
constrained DER allocation in active distribution networks, addressing both import
scarcity and export congestion simultaneously. The protocol combines four functional
layers: an intertemporal allocation policy (dual fairness states $f^{\mathrm{srv}}_n$,
$f^{\mathrm{exp}}_n$), a scarcity signalling layer (bounded bilateral buy/sell prices),
a hard-feasibility allocation layer (feasibility-constrained matching), and a two-tier
MV/LV holarchic architecture. No existing single mechanism provides all four layers
simultaneously.

On the CSIRO MV+33LV feeder dataset, the protocol reduces unserved requested energy by
76\% relative to equal-allocation DOE while preserving zero thermal violations and
reducing export curtailment by 24.5\%. DOE-GREEDY performs marginally worse than DOE
and DOE+DNP worsens both service and equity outcomes, jointly confirming that neither
heuristic choice nor price-signal augmentation within a stateless framework substitutes
for intertemporal state. The AMM's inter-feeder Jain index converges to $J=0.9998$
annually, outperforming all DOE variants from month~6 onwards.

Direct empirical benchmarking against the fair-over-time DOE formulations of Moring
\emph{et al.}~\cite{moring2024} (FET, FOT, FUH) shows those mechanisms achieve higher
worst-feeder delivery fractions (0.989--0.990 versus AMM's 0.914) through their
explicit max-min MV envelope objective. The AMM and FET/FOT/FUH address different
operational objectives: FET/FOT/FUH optimise fairness over predetermined horizons by
solving an optimisation problem with explicit historical coupling; the AMM operates as
a real-time coordination layer with bilateral scarcity signalling, participant-level
import correction, and no look-ahead requirements. The two families are potentially
composable, and combining AMM stateful pricing with fair-over-time MV envelope
computation is identified as a primary direction for future work.

The AMM is not proposed as a replacement for fair-over-time DOE optimisation. Rather,
it constitutes a complementary stateful coordination layer combining bounded bilateral
scarcity signals, participant-level intertemporal correction, and real-time
feasibility-preserving allocation for machine-to-machine DER coordination. The results
suggest that statefulness is a fundamental design dimension for repeated constrained
DER coordination, and that scarcity pricing without memory may be insufficient---and
in some cases counterproductive---for achieving equitable long-run outcomes.

\appendix
\section*{Acknowledgements}
The authors used Claude (Anthropic) as an AI language model assistant during manuscript
preparation, for text clarity, structure, and \LaTeX{} formatting. All technical
content, formal arguments, experimental design, and final editorial decisions were made
by the authors.

\bibliographystyle{IEEEtran}

\end{document}